\documentclass[a4paper,USenglish]{lipics-v2019}

\nolinenumbers
\hideLIPIcs

\long\def\commentbegin #1\commentend{}

\usepackage{balance}
\usepackage{algorithm}
\usepackage[noend]{algpseudocode}
\usepackage{ifthen}
\usepackage{comment}
\usepackage{amsthm,amsmath,amssymb}
\usepackage{graphicx}

\usepackage{pstricks,pst-node,pst-tree}
\usepackage{soul}

\usepackage{color}
\usepackage{mathrsfs }
\usepackage[normalem]{ulem}
\usepackage{paralist}

\renewcommand{\epsilon}{\varepsilon}
\newcommand{\Exp}[1]{\hbox{{\rm I}\hskip -2pt {\rm E}}\left[#1\right]}
\newcommand{\Prob}[1]{\hbox{\rm I\kern-2pt P}\left[#1\right]}

\DeclareMathAlphabet{\mathsc}{OT1}{cmr}{m}{sc}

\newtheorem{observation}{Observation}

\renewcommand{\geq}{\geqslant}
\renewcommand{\ge}{\geqslant}
\renewcommand{\leq}{\leqslant}
\renewcommand{\le}{\leqslant}

\newboolean{short}
\setboolean{short}{false}

\algblockdefx[ExecBl]{BlockOn}{BlockOff}  [1]{#1}
  
\makeatletter
\ifthenelse{\equal{\ALG@noend}{t}}%
  {\algtext*{BlockOff}}
  {}%
\makeatother

\newcommand{\shortOnly}[1]{\ifthenelse{\boolean{short}}{#1}{}}
\newcommand{\onlyShort}[1]{\ifthenelse{\boolean{short}}{}{#1}}
\newcommand{\longOnly}[1]{\ifthenelse{\boolean{short}}{}{#1}}
\newcommand{\onlyLong}[1]{\ifthenelse{\boolean{short}}{}{#1}}

\ifdefined\ShowComment
\def\billy#1{{\color{green}\underline{\textsf{Billy:}}} {\color{blue} \emph{#1}}}
\def\gopal#1{{\color{red}\underline{\textsf{Gopal:}}} {\color{blue} \emph{#1}}}
\def\david#1{{\color{orange}\underline{\textsf{David:}}} {\color{blue} \emph{#1}}}
\def\shay#1{\hl{shay: #1}}
\else
\def\billy#1{}
\def\gopal#1{}
\def\david#1{}
\def\shay#1{}
\fi

\def\inline#1:{\par\vskip 7pt\noindent{\bf #1:}\hskip 10pt}

\def\candidateSTATE{\textsf{Candidate}}
\def\REQUEST{\textsc{request}}
\def\APPROVED{\textsc{approved}}
\def\WINS{\textsc{wins}}

\def\electedSTATE{\textsf{Elected}}
\def\LEADER{\textsc{leader}}

\def\nonelectedSTATE{\textsf{Non-elected}}
\def\DECLINED{\textsc{declined}}
\def\LOSES{\textsc{loses}}
\def\DECIDE{\textsc{decide}}

\def\numphaserand{K}
\def\Chosen{\textsf{Chosen}}
\def\Contender{\textsf{Contender}}
\def\CANDSTATE{\textsf{CAND-STATE}}
\def\REFSTATE{\textsf{REF-STATE}}
\def\ID{RANK}
\def\PH{PH}
\def\RECORD{{\mathcal{P}}}

\def\hht{{\hat t}}
\def\WINNER{\textsc{Winner}}

\renewcommand{\geq}{\geqslant}
\renewcommand{\ge}{\geqslant}
\renewcommand{\leq}{\leqslant}
\renewcommand{\le}{\leqslant}

\long\def\hide #1\hideend{}

\newcommand{\squishlist}{
 \begin{list}{$\bullet$}
  { \setlength{\itemsep}{0pt}
     \setlength{\parsep}{3pt}
     \setlength{\topsep}{3pt}
     \setlength{\partopsep}{0pt}
     \setlength{\leftmargin}{1.5em}
     \setlength{\labelwidth}{1em}
     \setlength{\labelsep}{0.5em} } }

\newcommand{\squishlisttwo}{
 \begin{list}{$\bullet$}
  { \setlength{\itemsep}{0pt}
     \setlength{\parsep}{0pt}
    \setlength{\topsep}{0pt}
    \setlength{\partopsep}{0pt}
    \setlength{\leftmargin}{2em}
    \setlength{\labelwidth}{1.5em}
    \setlength{\labelsep}{0.5em} } }

\newcommand{\squishend}{
  \end{list}  }

\title{Singularly Optimal Randomized Leader Election}

\author{Shay Kutten}{Faculty of Industrial Engineering and Management, Technion - Israel Institute of Technology, Haifa, Israel}{kutten@technion.ac.il}{0000-0003-2062-6855}{This work was supported in part by the Bi-national Science Foundation (BSF) grant 2016419.}

\author{William K. Moses Jr.}{Faculty of Industrial Engineering and Management, Technion - Israel Institute of Technology, Haifa, Israel}{wkmjr3@gmail.com}{0000-0002-4533-7593}{This work was supported in part by the BSF grant 2016419 and in part by a Technion fellowship.}

\author{Gopal Pandurangan}{Department of Computer Science, University of Houston, Houston, TX, USA}{gopal@cs.uh.edu}{0000-0001-5833-6592}{G. Pandurangan was supported, in part, by NSF grants CCF-1527867, CCF-1540512,  IIS-1633720,  CCF-1717075, and BSF grant 2016419.}

\author{David Peleg}{Department of Computer Science and Applied Mathematics, Weizmann Institute of Science, Rehovot, Israel}{david.peleg@weizmann.ac.il}{0000-0003-1590-0506}{Supported in part by the US-Israel Binational Science Foundation grant 2016732.}

\authorrunning{S. Kutten, W.\,K. Moses Jr., G. Pandurangan, and D. Peleg}

\Copyright{Shay Kutten, William K. Moses Jr., Gopal Pandurangan, and David Peleg}



\EventEditors{Hagit Attiya}
\EventNoEds{1}
\EventLongTitle{34rd International Symposium on Distributed Computing (DISC 2020)}
\EventShortTitle{DISC 2020}
\EventAcronym{DISC}
\EventYear{2020}
\EventDate{October 12--18, 2020}
\EventLocation{Virtual Conference}
\EventLogo{}
\SeriesVolume{179}
\ArticleNo{17}

\begin{document}
\maketitle

\begin{abstract}
This paper concerns designing  distributed algorithms that are {\em singularly optimal}, i.e., algorithms that are {\em simultaneously}  time and message {\em optimal}, for  the  fundamental  leader election problem in networks. Our main result is a randomized distributed leader election algorithm for {\em asynchronous complete} networks  that is essentially (up to a polylogarithmic factor) singularly optimal.  Our   algorithm  uses $O(n)$  messages with high probability\footnote{Throughout, ``with high probability'' means with probability at least $1-1/n^c$, for a constant $c > 0$.} and runs in $O(\log^2 n)$ time (with high probability) to elect a unique leader. The $O(n)$ message complexity should be contrasted with the $\Omega(n \log n)$ lower bounds for the deterministic message complexity of leader election algorithms (regardless of time), proven by Korach, Moran, and Zaks (TCS, 1989) for asynchronous algorithms and by Afek and Gafni (SIAM J. Comput., 1991) for synchronous networks. Hence, our result also separates the message complexities of randomized and deterministic  leader election. More importantly, our (randomized) time complexity of $O(\log^2 n)$ for obtaining the optimal $O(n)$ message complexity is significantly smaller than the long-standing $\tilde{\Theta}(n)$ time complexity obtained by Afek and Gafni and by Singh (SIAM J. Comput., 1997)  for message optimal (deterministic) election in asynchronous networks. Afek and Gafni also  conjectured that $\tilde{\Theta}(n)$ time would be optimal for message-optimal  asynchronous algorithms. Our result shows that randomized algorithms are significantly faster. 

Turning to {\em synchronous} complete networks, Afek and Gafni showed an essentially singularly optimal deterministic algorithm with $O(\log n)$ time and $O(n \log n)$ messages. Ramanathan et al. (Distrib. Comput. 2007) used randomization to improve the message complexity, and showed a randomized algorithm with $O(n)$ messages but still with $O(\log n)$ time (with failure probability $O(1 / \log^{\Omega(1)}n)$). Our second result shows that synchronous complete networks admit a {\em tightly} singularly optimal randomized algorithm, with $O(1)$ time and $O(n)$ messages (both bounds are optimal). Moreover, our algorithm's time bound holds with certainty, and its message bound holds with high probability, i.e., $1-1/n^c$ for constant $c$. 

Our results demonstrate that leader election can be solved in a simultaneously message and time-efficient manner  in asynchronous complete networks using randomization. It is open whether  this is possible in  asynchronous general networks.    
\end{abstract}

\section{Introduction}
\label{sec:intro}

Leader election is a classical and  fundamental problem in distributed computing with numerous applications; see e.g., \cite{Lann77-DistSystems, GallagerHS1983,AG91,humblet-clique,KorachKuttenMoran-ModularLE-TOPLAS,KorachOptimalTrees87,KorachOptimal89,KPPRT15tcs,singh,KPPRT15jacm,sense-of-dir,LouiMW88,singh97,KhanKMPT08, Lyn96,peleg-jpdc,santoro-book,GerardTelDistributedAlgosBook,KorachPODC1984,Pet84,AngluinSTOC80,KrishnaRamanathan:randomized,Gupta-ProbLE-DISC,dmitry,AM94}. 
The goal is to select a unique node, called the {\em leader}, from a set of nodes. An {\em arbitrary} subset of nodes can {\em wake up spontaneously at arbitrary times} and start the election algorithm by sending messages over the network. When the algorithm terminates, a {\em unique} node $v$ must be elected as leader and be {\em known to all nodes}.

Election is important both theoretically and because of the multiple applications such as implementing databases and data centers \cite{isard2007autopilot,chandra2007paxos,TV07}, locks \cite{27897}, file servers \cite{ghemawat2003google,chang2008bigtable}, broadcast and multicast \cite{perlman2000interconnections,castro2002scribe},
 and virtually every global task \cite{awerbuch1990communication}. In many of these applications, the network is treated
 as being virtually complete.
With the advent of large-scale  and resource-constrained networks such as peer-to-peer systems~\cite{KSSV06,KSS06,soda12,Ratnasamy01CAN,Rowstron01Pastry,Zhao04Tapestry} and ad hoc and sensor networks (e.g., \cite{sensor2,sensor}), it is often desirable to achieve low cost and scalable leader election.  Leader election has been studied extensively over the years in various distributed computing models  starting with the standard CONGEST model of networks (in particular, ring, complete networks, and arbitrary networks), radio networks, peer-to-peer networks,  population protocols, and programmable matter to name a few.

Our goal in this paper is to design  leader election algorithms in distributed networks
such that each is efficient with respect to {\em both of}
the two fundamental complexity measures, namely {\em messages} and {\em time}. 
Unfortunately, designing distributed network algorithms that are {\em simultaneously} time- and message-efficient has proved to be a challenging  task.
Consequently, research in the last three decades  has focused mainly on optimizing either one of the two measures 
separately, typically at the cost of neglecting the other.  There has been significant recent progress in obtaining algorithms that are essentially optimal in both measures (or at least work well under both measures  to the extent possible) for various problems such as leader election, minimum spanning tree and shortest paths \cite{KPPRT15jacm,PRS17,elkin2017simple, haeupler2018round}.
In particular, as defined in \cite{PRS17} (see also \cite{gmyr}), the following two notions are of interest for any given problem:

\noindent  --- {\bf Singular optimality:} 
A problem enjoys {\em singular optimality} if it has a distributed algorithm that is optimal (or at least optimal up to a polylogarithmic factor) with respect to both measures simultaneously.
\\
\noindent --- {\bf Time-message trade-off:} 
A problem exhibits a {\em time-message trade-off} if it fails to admit a singularly optimal solution, namely, algorithms of better time complexity for it necessarily incur higher message complexity and vice versa. 

The singularly optimal results mentioned earlier for leader election, minimum spanning tree,
and shortest paths \cite{KPPRT15jacm,PRS17,elkin2017simple,haeupler2018round}, crucially
apply only to {\em synchronous} networks.
A main motivation for this work is to study whether leader election admits singularly optimal algorithms or exhibits a time-message tradeoff in {\em asynchronous}  networks. 
Unlike synchronous  networks which admit a leader election algorithm  which is singularly optimal (up to a logarithmic factor) \cite{KPPRT15jacm}, it is not known whether  asynchronous networks admit such an algorithm (see ``Background and Prior Work'').

In this paper, we focus on the classical problem of leader election in {\em complete} networks, which itself has been studied extensively for nearly four decades.
In such a network, every pair of nodes is connected by a bidirectional communication link and in the beginning, no node has any knowledge
of any other nodes, including their identities (if any). Like many prior works (e.g., \cite{AG91, singh, humblet-clique,GallagerHS1983,AM94} and others), this paper focuses on the more challenging {\em asynchronous} model and provides the first-known singularly optimal algorithm. 

It is clear that the best possible time bound for leader election in complete networks is {\em constant} if there is no restriction on the message complexity. The situation with respect to the message complexity is more nuanced. 
(The message lower bound is trivially $\Omega(n)$ since the leader has to be known to all nodes.) 
For {\em deterministic} algorithms, there is a well-known lower bound of
$\Omega(n \log n)$ messages (even in the synchronous setting when an adversary decides which nodes to wake up and when) \cite{AG91,KorachOptimal89,KorachOptimalTrees87} where $n$ is the number of network nodes. 
Moreover, Afek and Gafni \cite{AG91}
show that any deterministic algorithm that is message optimal (i.e., takes $\Theta(n \log n)$ messages) needs $\Omega(\log n)$ time; they present such a deterministic algorithm that takes $O(n \log n)$ messages and $O(\log n)$ time in {\em synchronous} networks. Hence for synchronous networks, there exists a deterministic algorithm that is essentially singularly optimal (up to a logarithmic factor).

The situation is  less clear in the asynchronous setting and for randomized solutions. First, the $\Omega(n \log n)$ message lower bound was shown for deterministic algorithms, and it was not clear whether by using randomization one can breach this bound. Second, Afek and Gafni~\cite{AG91} present an {\em asynchronous} leader election algorithm that takes $O(n \log n)$ messages (which is message optimal), but takes
$O(n)$ time. They conjecture that this is the best possible time bound for message-optimal asynchronous algorithms. Singh \cite{singh} (mildly) disproves this conjecture by presenting
an algorithm that runs in $O(n/\log n)$ time and is also message optimal ($O(n\log n)$ messages). This is still a far cry from being singularly optimal, as the time bound is essentially linear.  When it comes to the use of randomness, Afek and Matias~\cite{AM94} develop a randomized algorithm that succeeds with probability $1 - \epsilon$ and runs in $O(\log n)$ time using $O((n/\epsilon)\log^2 (1/\epsilon))$ messages.\footnote{In the paper, it is claimed that they develop an algorithm with termination detection that runs in $O(\log n)$ time for complete graphs. However, it is unclear if indeed the running time is $O(\log n)$ for an algorithm with termination detection or rather $O(n)$ time when termination detection is required.} If we allow for a constant probability of success, i.e., if we set $\epsilon$ to a fixed constant, then the algorithm runs in $O(\log n)$ rounds and uses $O(n)$ messages. However, for randomized algorithms, it is typically desired that algorithms succeed with high probability. In such a scenario if we set $\epsilon$ to $1/n^c$ for some constant $c>0$, the message complexity significantly increases to $O(n^{1+c} \log^2 n)$.

The results of the current paper (detailed later on) {\em break} the long-standing barrier of $\Omega(n \log n)$ messages for deterministic algorithms and show that there exist singularly optimal randomized asynchronous algorithms that succeed with high probability. 

\smallskip 
\noindent \textbf{Distributed Computing Model.} 
We consider a system represented as an undirected complete graph $G=(V,E)$, $|V|=n$, similar
to the models of~\cite{AG91,humblet-clique,KorachKuttenMoran-ModularLE-TOPLAS,KorachOptimalTrees87,KorachOptimal89}, except that processors can access {\em private unbiased coins}.
Our upper bounds do not require unique identities, and in particular, nodes can be anonymous. 
If nodes do have unique identifiers, then we assume, as in prior works on complete networks (see \cite{AG91}), that nodes initially do not know the unique identifiers of other nodes.\footnote{Otherwise (in the $KT_1$ model \cite{AGVP90}, where each node knows the unique identifiers of other nodes), leader election is trivial in complete networks.}

In the synchronous communication setting, the computation is divided into discrete lockstep time units called {\em rounds}, and every message sent over an edge arrives at the receiver after a fixed delay of one time unit, namely, at the end of the current round.
In contrast, in the {\em asynchronous} communication setting, 
messages sent on edges incur unpredictable but finite delay, in an error-free and FIFO manner (i.e., messages will arrive in sequence).
Nevertheless, for the sake of time analysis it is assumed that a message  takes {\em at most one time unit} to be delivered.
For both modes, we make the usual (and here, very reasonable) assumption that local computation within a node is  instantaneous and free.
We assume the  standard $\mathcal{CONGEST}$ model~\cite{Pel00}, where a node can send at most one $O(\log n)$ bit message on each edge in each round.

Following the standard assumption in asynchronous protocols (see \cite{AG91,GallagerHS1983,singh}),
nodes are initially asleep.
A node enters the execution when it is woken up by the environment (at most once), or upon receiving messages from awakened neighbors.
In the asynchronous setting, once a node enters execution, it performs all the computations required of it by the algorithm, sends out messages to neighbors as specified by the algorithm. 
In the synchronous setting, we assume the above {\em adversarial} wake up assumption as well.\footnote{This should be contrasted with the {\em simultaneous wake up} model, where all nodes are assumed to be awake at the beginning of computation; this is typically assumed in design of synchronous protocols (see e.g., \cite{AG91,KPPRT15tcs,KPPRT15jacm}).}

The message complexity of an algorithm is the worst-case total number of $O(\log n)$ bit messages sent during its execution. 
The time complexity is the worst-case total number of time units since the first node is woken up to the last message transmission due to the algorithm. 
Note that a time unit in the synchronous model corresponds to one round, whereas in the asynchronous model it is an upper bound on the transmission time of a message over an edge.
Hence, it is generally more difficult to design efficient algorithms in asynchronous systems than in synchronous systems.

Following the standard approach to modeling distributed networks (see \cite{AG91}), we represent the environmental uncertainties by means of an {\em adversary} controlling some of the execution parameters. Specifically, we assume an {\em adversarial wake up} mode, where node wake-up is scheduled by an adversary  
(who may decide to keep some nodes dormant).
A node can also be woken up by receiving messages from other nodes. 
In addition to the wake-up schedule, the adversary also decides for how long to delay each message.
These decisions are done {\em adaptively}, i.e., when the adversary makes a decision to wake up a node or delay a message, it has access to the results of all previous coin flips. 
Finally, recalling that initially the nodes are unaware of which neighbor is connected to each of their outgoing edges, the adversary also controls the graph structure (i.e., the mapping of outgoing edges to neighbors).
Here, we consider the adversary to be {\em oblivious}, i.e., if nodes use randomness to choose an outgoing edge, the adversary must choose endpoints for all such outgoing edges prior to the first use of randomness in this way.\footnote{To appreciate the implications of allowing the adversary to construct the graph {\em after} seeing
the random choices, see the lower bound proof techniques for message complexity in~\cite{KorachOptimal89,AG91}. These techniques are for deterministic algorithms, but will work when the adversary has adaptive edge mapping abilities.}

\smallskip 
\noindent \textbf{Background and additional Prior Work.}
The complexity of the leader election problem and (especially deterministic) algorithms for it
have been very well-studied in distributed networks. Various algorithms and lower bounds are known 
in different models with synchronous/asynchronous communication and 
in networks of varying network topologies such as a cycle, a complete graph,
 or some arbitrary topology (e.g., see
\cite{KhanKMPT08,KorachKuttenMoran-ModularLE-TOPLAS,KPPRT15jacm,KPPRT15tcs, Lyn96,peleg-jpdc,santoro-book,GerardTelDistributedAlgosBook} 
and the references therein). The problem was first studied in context of a ring network by Le~Lann~\cite{Lann77-DistSystems} and discussed for general graphs in the influential paper of Gallager, Humblet, and Spira~\cite{GallagerHS1983}. Kutten et al. \cite{KPPRT15jacm} presented a singularly optimal (up to a logarithmic factor) randomized leader election algorithm for general {\em synchronous} networks that ran in $O(D)$ time and used $O(m \log n)$ messages (where $D$, $m$, and $n$ are the network diameter, number of edges, and the number of nodes respectively). We note that $\Omega(D)$ and $\Omega(m)$
are lower bounds for time and messages for leader election even for
randomized algorithms \cite{KPPRT15jacm}. It is not known whether similar bounds can be achieved for general {\em asynchronous} networks, although one can obtain algorithms that are separately time optimal \cite{peleg-jpdc} and message optimal \cite{GallagerHS1983}.

Leader election in the class of 
{\em complete networks} --- which is the focus of this paper --- has come to occupy a special position of its own 
and has been extensively studied~\cite{AG91,AM94,humblet-clique,KorachKuttenMoran-ModularLE-TOPLAS,KorachOptimalTrees87,KorachOptimal89,KPPRT15tcs,singh}; see also \cite{sense-of-dir,LouiMW88,singh97} for leader election in complete networks where nodes have a sense of direction. While $\Omega(n)$
is an obvious lower bound on the message complexity of leader election when the leader's identity should be known for all nodes,
an $\tilde{O}(\sqrt{n})$ (i.e., sublinear) message complexity can be obtained for a related but different problem in {\em synchronous} complete networks with {\em simultaneous
wake up} \cite{KPPRT15tcs} where we {\em do not} require that 
the nodes not elected know who is the leader (nor which of their ports lead to it).\footnote{This variant of leader election is sometimes called ``implicit'', as opposed
to the version studied here, where all nodes need to know the identity of the leader.}
The above result  crucially uses randomization to break the linear message complexity
bound that applies  for deterministic algorithms.

The study of leader election algorithms is usually concerned with both message and time complexity. 
Korach et al.~\cite{KorachPODC1984}, Humblet \cite{humblet-clique},  Peterson \cite{Pet84} and  Afek and Gafni~\cite{AG91} presented $O(n \log n)$ message algorithms for {\em asynchronous} complete networks.  Korach, Kutten, and Moran~\cite{KorachKuttenMoran-ModularLE-TOPLAS} presented a general method plus applications to various classes of graphs including complete networks.  Afek and Matias~\cite{AM94} similarly presented a general method with application to a complete graph.

Afek and Gafni (as a part of presenting a tradeoff between time and message complexity) showed that the time complexity of a message optimal {\em synchronous} algorithm was $\Theta(\log n)$, while for a message optimal {\em asynchronous} algorithm, they only demonstrated an $O(n)$ time upper bound (improving previous time bounds for message optimal algorithms). They conjectured that {\em in the  asynchronous case, time complexity of any message optimal algorithm is} $\Omega(n)$. Singh \cite{singh} (as a part of presenting a different tradeoff), presented a somewhat better ($O(n / \log n)$) time for message optimal asynchronous algorithms, but still posed as an important open problem the question whether $\Omega(n / \log n)$ time was optimal for such algorithms. 
The current paper demonstrates that this is not the case, at least for randomized algorithms that succeed w.h.p., by giving a $O(\log^2 n)$ time bound for message optimal $O(n)$ algorithm. 

The earlier mentioned papers~\cite{sense-of-dir,LouiMW88,singh97} on ``{\em sense of direction}'' demonstrated that when the nodes possessed additional knowledge on the topology, it was possible to reduce the number of messages to $O(n)$.
In the same vein, note that for {\em deterministic} algorithms in  {\em  synchronous} 
networks and under the strong assumption of {\em simultaneous wake up}, there exists an $O(n)$ messages algorithm \cite{KPPRT15jacm}.
In contrast, an $\Omega(n\log n)$ message lower 
bound in synchronous networks was shown by Afek and Gafni~\cite{AG91} under the more common assumption of adversarial wake up.
The message complexity in the current paper is $O(n)$ without using sense of direction or simultaneous wakeup assumptions (but using randomized algorithms).

For anonymous networks under some reasonable assumptions, deterministic leader election was shown to be impossible, using symmetry arguments~\cite{AngluinSTOC80}. Randomization comes to the rescue in this case; random rank assignment is often used to assign unique identifiers, as done herein. Randomization also allows us to beat the lower bounds for deterministic algorithms, albeit at the risk of a small chance of error. For example, Afek and Matias~\cite{AM94} developed a randomized algorithm that succeeded with probability $1 - \epsilon$ and ran in $O(\log n)$ time using $O((n /\epsilon)\log^2 (1/\epsilon))$ messages. It should be noted that Singh's~\cite{singh} deterministic algorithm allowed a trade-off between time and messages such that it was possible to achieve a running time as low as $O(\log n)$ in exchange for more messages ($O(n^2/\log n)$). By setting $\epsilon$ to $1/n^c$, for a constant $c>0$, we see that Afek and Matias's algorithm succeeds with high probability and takes less messages ($O(n^{1+c} \log^2 n)$) for the same running time.

Turning to {\em synchronous} complete networks, Afek and Gafni showed an essentially singularly optimal deterministic algorithm with $O(\log n)$ time and $O(n \log n)$ messages.
Ramanathan et al. used randomization to improve the message complexity, and showed a randomized leader election algorithm for synchronous networks that could 
err with probability $O(1 / \log^{\Omega(1)}n )$ 
with time $O(\log n)$ and $O(n)$ messages\footnote{In contrast, the synchronous algorithm presented in the current paper succeeds with high probability, its time complexity is constant, and its complexity bounds hold with probability $1-O(1 / n^{\Omega(1)})$.}~\cite{KrishnaRamanathan:randomized}.
That paper also extends the synchronous algorithm to work for {\em partially} synchronous
networks, where message delays are bounded. 
It also surveys some related papers about randomized algorithms in 
other models that use more messages for performing leader election
~\cite{Gupta-ProbLE-DISC} or related tasks (e.g., quorum systems,
 Malkhi et al.~\cite{Malkhi-ProbQSystems}).
    In the context of self-stabilization, a randomized
 algorithm with $O(n \log n)$ messages and $O(\log n)$ time until stabilization was presented in \cite{dmitry}.

\begin{table*}[ht]
	\caption{\small Comparison of previous upper bound results for leader election in complete networks of $n$ nodes along with our contributions. 
	The variables $2 \leq c \leq n$, $\log n \leq k \leq n$, and $0 < \epsilon < 1$ are parameters provided to the respective algorithms. 
	} 
	\centering 
		\resizebox{1.0\columnwidth}{!}{%
	\begin{tabular}{|c|c|c|c|l|}
		\hline
		Paper & Message Complexity & Time Complexity  & Communication  & Type of \\
		& &    & Mode  & Solution\\
		\hline
		\hline
		\cite{KorachPODC1984} & $O(n \log n)$ & $O(n \log n)$ & Asynchronous  & Deterministic\\
		\hline
		\cite{AG91} & $O(n \log n)$ & $O(n)$ & Asynchronous & Deterministic  \\
		\hline
		\cite{singh} & $O(nk)$ & $O(n/k)$ & Asynchronous & Deterministic\\
		\hline
		\cite{AM94} & $O(n \log^2 (1/\epsilon) /\epsilon)$ & $O(\log n)$ & Asynchronous & Randomized, success prob. $1 - \epsilon$ \\
		\hline
		This paper & $O(n)$ w.h.p.  & $O(\log^2 n)$ w.h.p.  & Asynchronous & Randomized, success  w.h.p.*\\
		\hline
		\hline
		\cite{AG91} & $O(c n \log_c n)$ & $O(\log_c n)$ & Synchronous & Deterministic  \\
		\hline
		\cite{KrishnaRamanathan:randomized} & $O(n)$ & $O(\log n)$ & Synchronous & Randomized, success prob. $1 - O(1 / \log^{\Omega(1)}n)$ \\
		\hline
		This paper & $O(n)$ w.h.p. & $O(1)$  & Synchronous & Randomized, success w.h.p.*\\
		\hline
		\hline
		\multicolumn{5}{|l|}{*The algorithm always succeeds when nodes have unique identifiers (as opposed to anonymous networks).}\\
		\hline
	\end{tabular}
		}
	\label{table:upper-bound-results}
\end{table*}

\smallskip 
\noindent \textbf{Our Main Results.}
The main focus of this paper is on studying how randomization can help in designing singularly optimal algorithms for leader election in asynchronous as well as synchronous networks. 
Our results are summarized in Table 1.
Our main result is a  randomized asynchronous  leader election algorithm  for complete networks that runs in $O(\log^2 n)$ time and uses only $O(n)$
messages to elect a unique leader with high probability (Section \ref{sec:adv-wakeup-asynch}). 
This is a significant improvement over the $\Omega(n \log n)$ messages needed for any deterministic algorithm, and an even larger improvement over the time complexity of previous message optimal ($O(n \log n)$) deterministic algorithms for asynchronous networks, which required $O(n)$ or $O(n/\log n)$ time \cite{AG91,singh}.
In addition, we show that for the synchronous setting too, we can obtain a  
singularly optimal algorithm that tightly matches the best possible asymptotic time and message lower bounds. We present a randomized algorithm in synchronous networks that takes $O(1)$ time and $O(n)$ messages with high probability (see Section \ref{sec:adv-wakeup-synch-alg}). Note that our algorithms succeed w.h.p. in anonymous networks and always succeed when each node has a unique identifier.

Our results, while providing near-optimal message and time bounds for complete asynchronous networks, are a step towards designing singularly optimal algorithms for {\em general} asynchronous networks.

\section{A Randomized Algorithm for Asynchronous Networks \& Analysis}
\label{sec:adv-wakeup-asynch}

In this section, we first give some high level intuition of how we use  randomization to overcome the barriers of time and message complexity posed to deterministic asynchronous algorithms. Subsequently, we present a randomized algorithm that solves leader election in $O(\log^2 n)$ time with high probability using $O(n)$ messages with high probability.

\smallskip 
\noindent \textbf{High level intuition behind the use of randomness.} In~\cite{AG91}, an $\Omega(n \log n)$ message lower bound is presented. A key idea in the proof is that the adversary is able to control the destination of messages sent by a node. It can do this for a deterministic algorithm because the adversary may predict which edges will be used in the course of the algorithm and construct the initial graph accordingly. Intuitively, if multiple clusters of nodes are active, a leader cannot be determined until they interact. If the adversary can consistently control the destination of messages over unknown links, the adversary may blow up the number of messages until the clusters interact. We bypass this issue by having nodes choose the edges they use to communicate at random, similar to the idea used in~\cite{KPPRT15tcs}. Since the adversary cannot predict which edge will be used, it cannot initially construct an undesirable graph. Our second use of randomness is to have each node, once awake, choose its $\ID$ uniformly at random from $[1,n^4]$. This helps us achieve a significantly  smaller running time compared to the deterministic case. If $\ID$s are not randomly chosen but deterministically assigned, then the adversary may wake up nodes such that it takes $O(n)$ time from the time the first node is woken up until the algorithm terminates. This is because, in our algorithm, a node's $\ID$ plays an important role in deciding if it will go on to become the leader when interacting with another node or if it will stop trying to be the leader.

\smallskip 
\noindent \textbf{The Randomized algorithm.}
We now describe the algorithm. Each awake node may play up to two roles in the algorithm, namely (i) a candidate and (ii) a referee. 
During the process, candidates attempt to progress towards becoming the leader. The position of a candidate $u$ in this process is represented by the pair $\langle\ID_u,\PH_u\rangle$, where $\ID_u$ is $u$'s {\em rank} and $\PH_u$ is its current {\em phase}. Intuitively, we say that candidate $v$ is {\em ahead of} candidate $u$, and $u$ is {\em behind} $v$, if $v$ has a higher phase number, or, in the case of a tie, $v$ has a higher \ID. Formally, we say that $v\gg u$ if either (i) $\PH_v > \PH_u$, or (ii) $\PH_v=\PH_u$ and $\ID_v > \ID_u$.\footnote{If nodes have unique identifiers, then in the case of two nodes in the same phase with the same rank, their unique identifiers can be used to break ties. Furthermore, a node will append its unique identifier to its $\ID$ to avoid ambiguity.}

During the process, candidates gradually either {\em progress} or \textit{retire}, eventually resulting in a single un-retired candidate who then becomes the leader. 
Essentially, a candidate $u$ retires upon encountering another candidate $v$ that is ahead of it. Such an encounter occurs when a referee learns of both candidates, and it is the referee's task to make one of the candidates retire.
(The referee is typically a third party, but may also be $u$ or $v$.)
Hence, candidates fundamentally drive the algorithm forward, while referees serve a complementary role of guardians, assisting candidates in deciding whether to retire or to proceed. Each node is in one of three candidate states $\candidateSTATE$, $\nonelectedSTATE$, $\electedSTATE$, stored in the variable $\CANDSTATE$, corresponding to whether the node is still a candidate in the running to become a leader, is retired, or is elected, respectively.

If a node is woken up by the adversary, then it plays the role of a candidate and may eventually become the leader. If a node is woken up by a message from a neighbor, then it is immediately considered to be a retired candidate. When a node receives a message, depending on the type of message received, the node may either act in the role of a candidate or in the role of a referee, with an appropriate procedure called to handle the message. Thus, during the execution, a node either acts as just a referee or as both a candidate and a referee.

Each node $u$, once woken up, checks if it was woken by the adversary (i.e., spontaneously wakes up itself) or by a message from a neighbor. If $u$ was woken by a neighbor, $u$ sets its candidate state to $\nonelectedSTATE$. If $u$ was woken by the adversary, $u$ sets its candidate state to $\candidateSTATE$, chooses an integer in $[1,n^4]$ uniformly at random as its $ID_u$ and attempts to progress. The initialization is described in Algorithm~\ref{alg:initialization}.

\alglanguage{pseudocode}
\begin{algorithm}
	\caption{Algorithm Initialization, run by each node $u$ upon waking up.}
	\label{alg:initialization}
	
	\begin{algorithmic}[1]
	\State If at any time $u$ receives a message $\langle \ID_v,0,\LEADER\rangle $ from another node $v$, then $u$ sets the leader as $v$, sets $\CANDSTATE \gets \nonelectedSTATE$, and terminates
    \State $\REFSTATE \gets C0$
    \BlockOn{Run Procedures \texttt{Candidate}, \texttt{Referee\_Request\_Response}, \texttt{Referee\_Dispute\_Response}, and \texttt{Decide} in parallel until termination, according to the following rules:}
	\If{$u$ spontaneously woke up}
		\State $\CANDSTATE \gets \candidateSTATE$ 
		\State Choose integer in $[1,n^4]$ uniformly at random to be $u$'s rank $\ID_u$
		\State Run Procedure \texttt{Candidate}
	\Else 
		\State $\CANDSTATE \gets \nonelectedSTATE$
	\EndIf
	
	\State If $u$ receives a message of the form $\langle \cdot, \cdot, \REQUEST \rangle$, run Procedure \texttt{Referee\_Request\_Response}
	\State If $u$ receives a message of the form $\langle \cdot ,\cdot,\WINS\rangle \circ \langle \cdot,\cdot, \LOSES \rangle$ or $\langle \cdot ,\cdot,\LOSES\rangle \circ \langle \cdot,\cdot, \WINS \rangle$, run Procedure \texttt{Referee\_Dispute\_Response}
	\State If $u$ receives a message of the form $\langle \cdot, \cdot, \DECIDE\rangle $, run Procedure \texttt{Decide}
	
	 \BlockOff
	\end{algorithmic}
\end{algorithm}

A candidate $u$ progresses through at most $\numphaserand=\lceil \log \sqrt{4n \log n} \rceil +1$ phases until it either changes its candidate state to $\nonelectedSTATE$ or completes phase $\numphaserand$ and stays in candidate state $\candidateSTATE$, in which case it declares itself leader.
 At the beginning of each phase $i=1$ to $\numphaserand-1$, a candidate node $u$ chooses a set $\mathcal{S}$ of $\min \lbrace 10 \cdot 2^i, \lceil \sqrt{4n \log n} \rceil \rbrace$ nodes uniformly at random designated as its {\em referees} and seeks their approval to progress.\footnote{We assume that node $u$ treats itself as a referee as well in addition to the nodes of $\mathcal{S}$.} In phase $i = \numphaserand$, $u$ sets $\mathcal{S} \gets V$, i.e. all nodes of the network form set $\mathcal{S}$.
 Node $u$ sends a {\em request} message $\langle \ID_u,\PH_u, \REQUEST\rangle$ to each referee in $\mathcal{S}$.

 When node $u$ receives replies from all nodes in $\mathcal{S}$, it checks if any of those replies is a decline of the form $\langle \ID_u, \PH_u, \DECLINED\rangle$. If so, $u$ changes its candidate state to $\nonelectedSTATE$.\footnote{The candidate state of $u$ can also be changed to $\nonelectedSTATE$ as a result of processing a $\DECIDE$ message, described later.} In case $u$ received no declines, it may proceed. In particular, if $i < \numphaserand$, then $u$ increases $\PH_u$ to $i+1$ and starts the next phase, and if $i=\numphaserand$ and $u$ still retains candidate state $\candidateSTATE$, then $u$ declares itself as leader, namely, it changes its candidate state to $\electedSTATE$, broadcasts the final announcement $\langle \ID_u,0,\LEADER\rangle$ and terminates.
The above process is described in Procedure~\ref{alg:candidate}.

\alglanguage{pseudocode}
\begin{algorithm}
	\caption{Procedure Candidate, run by each candidate $u$.}
	\label{alg:candidate}
	
	\begin{algorithmic}[1]
	\For{$i \gets 1$ to $\numphaserand$}
	    \If{$i < \numphaserand$}
		    \State Choose a set $S$ of $\min \lbrace 10 \cdot 2^i, \lceil \sqrt{4n \log n} \rceil \rbrace$ neighbors uniformly at random
		\Else
		    \State Set $S \gets V$
		\EndIf
		\State Send message $\langle \ID_u,i,\REQUEST \rangle$ to each node in $S$
		\State Wait for replies from all referees in $S$. \\ \hspace{20pt}
		(* Other procedures may run in parallel and do other things in the meantime.*)
		\If{any reply is of the form $\langle \ID_u,\PH_u, \DECLINED\rangle $}
			\State $\CANDSTATE \gets \nonelectedSTATE$
		\EndIf
		\If{$\CANDSTATE = \nonelectedSTATE$}
			\State Exit procedure
		\EndIf
	\EndFor
	\If{$\CANDSTATE = \candidateSTATE$}
		\State $\CANDSTATE \gets \electedSTATE$
		\State Broadcast $\langle \ID_u,0,\LEADER\rangle $ and terminate
	\EndIf
	\end{algorithmic}
\end{algorithm}

Each referee $r$ helps candidates retire (until only one is left) by comparing candidate pairs, keeping the more advanced one, and instructing the other to retire. 
Referee $r$ can be in one of four referee states stored in the variable $\REFSTATE$.

\begin{itemize}
\item
$\REFSTATE=C0$ holds when $r$ has not been approached by any candidate yet.
\item
$\REFSTATE=C1$ holds when $r$ has one approved candidate, referred to as its {\em chosen} candidate, and has declined every other candidate that approached it so far.
A record containing the position of the chosen candidate $v$, namely, 
$\RECORD(v) = \langle \ID_v,\PH_v\rangle$, is kept in the variable $\Chosen$. 
\item
$\REFSTATE=C2$ holds when $r$ currently keeps track of two candidates:
the chosen $v$ (in the variable $\Chosen$) and a contender $w$ (in the variable $\Contender$), such that $w\gg v$, and all other candidates that approached $r$ so far were declined. Moreover, a {\em dispute} is currently in progress between $v$ and $w$. This state is typically reached when $r$ has a chosen candidate $v$ that was approved by it, and later it gets a request from another candidate $w$ such that $w\gg v$. In this situation, $r$ cannot decline $w$ (since it is ahead of its current chosen), but at the same time it cannot approve $w$, since it may be that $v$ has progressed in the meantime, and it is now ahead of $w$. To resolve this uncertainty, $r$ declares a dispute, and sends a $\DECIDE$ message  containing $w$'s position $\langle\ID_w,\PH_w\rangle$ to $v$, asking it to make a comparison between $w$ and itself (based on its current phase $\PH_v$). While waiting for $v$'s response, $r$ keeps a record containing the details of $w$, namely, $\RECORD(w) = \langle \ID_w, \PH_w \rangle$, in the variable $\Contender$.
\item
$\REFSTATE=C3$ is similar to $C2$, i.e., it holds when $r$ currently keeps track of a chosen $v$ and a contender $w$, all other candidates that approached $r$ so far were declined, and a dispute is currently in progress. The difference, however, is that the on-going dispute does not involve $w$. Rather, it is between $v$ and some {\em previous} contender $z$ such that $w\gg z\gg v$. This state is typically reached when a new candidate $w$ approaches $r$ while $r$ is in referee state $C2$ with a dispute in progress between a chosen candidate $v$ and a contender $z\gg v$, and $r$ discovers that $w\gg z$. This allows $r$ to decline the current contender $z$ immediately (since even if the outcome of the dispute favors $z$ over $v$, the new candidate $w$ is ahead of $z$, so $z$ must retire). Now $w$ takes $z$'s place as the contender. 
\end{itemize}

If $r$ receives a message  $\langle \ID_u,\PH_u, \REQUEST\rangle $ from node $u$ over some edge $e$, it responds as follows. 
\begin{itemize}
\item
If $r$ is in referee state $C0$, then it registers $u$ as its chosen candidate, sends back an approval message, and switches its referee state to $C1$.
\item
If $r$ is in referee state $C1$, then the following sub-cases may apply:
If the current chosen is also $u$ (from an earlier phase), then $r$ updates the record stored in $\Chosen$.
If it is another node $v$ such that $v\gg u$, then $r$ sends $u$ a decline message.
Otherwise (i.e., if $u\gg v$), $r$ registers $u$ as the contender and initiates a dispute by sending
a $\DECIDE$ message to $v$, requesting it to compare its current position with that of $u$. It also switches its referee state to $C2$. 
\item
If $r$ is in referee state $C2$, signifying that a dispute is in progress, then $r$ compares the new candidate $u$ with the current contender $w$. 
If $u\ll w$ then $r$ declines $u$'s request (and thus retires $u$). 
Otherwise (i.e., if $u\gg w$), $r$ retires $w$, registers $u$ as the new contender, and switches its referee state to $C3$.
\item
If $r$ is in referee state $C3$, then it does the same as in referee state $C2$.
\end{itemize}
The pseudocode for the referee's actions on receiving a request is given in Procedure~\ref{alg:referee}.

\alglanguage{pseudocode}
\begin{algorithm}
	\caption{Procedure Referee\_Request\_Response, run by each referee $r$ on receiving a message $\langle \ID_u, \PH_i, \REQUEST\rangle$ from candidate $u$.}
	\label{alg:referee}
	
\begin{algorithmic}[1]
\If{$\REFSTATE=C0$ (* $u$ is the first candidate to approach $r$ *)}
	\State Set $\Chosen\gets \RECORD(v) = \lbrace \ID_u, \PH_u \rbrace$
	\State Send message $\langle \ID_u, \PH_i, \APPROVED \rangle$ to $u$
	\State Set $\REFSTATE \gets C1$
\ElsIf{$\Chosen=u$ (* $u$ is the current chosen candidate, from an earlier phase *)}
	\State Set $\Chosen\gets \lbrace \ID_u, \PH_u \rbrace$
\ElsIf{$\REFSTATE=C1$ (* There is a chosen candidate $v=\Chosen$; all other candidates that approached $r$ so far were rejected *)}
    \If{$u \ll v$}
    \State Send message $\langle \ID_u, \PH_u, \DECLINED \rangle$ to $u$
    \Else
    \State Set $\Contender\gets \RECORD(u) = \lbrace \ID_u, \PH_u \rbrace$
    \State Send message $\langle\ID_u,\PH_u,\DECIDE\rangle$ to the chosen $v$
    \State Set $\REFSTATE \gets C2$
    \EndIf
\ElsIf{$\REFSTATE=C2$ (* A dispute is in progress between the current chosen $v=\Chosen$ and the current contender $w=\Contender$, $w\gg v$ *)}
	\If{$u\ll w$}
	\State Send message $\langle \ID_u, \PH_u, \DECLINED \rangle$ to $u$
	\Else
	\State Send message $\langle \ID_w, \PH_w, \DECLINED \rangle$ to $w$
	\State Set $\Contender\gets \RECORD(u) = \lbrace \ID_u, \PH_u \rbrace$
    \State Set $\REFSTATE \gets C3$
	\EndIf
\ElsIf{$\REFSTATE=C3$ (* A dispute is in progress between the current chosen $v=\Chosen$ and some predecessor $z$ of the current contender $w=\Contender$, $w\gg z\gg v$ *)}
    \If{$u\ll w$}
    \State Send message $\langle \ID_u, \PH_u, \DECLINED \rangle$ to $u$
    \Else
    \State Send message $\langle \ID_w, \PH_w, \DECLINED \rangle$ to $w$
	\State Set $\Contender\gets \RECORD(u) = \lbrace \ID_u, \PH_u \rbrace$
    \EndIf
\EndIf
\end{algorithmic}
\end{algorithm}

When a node $v$ which is currently the chosen candidate of the referee $r$ (but may have possibly retired since the time it was approved by $r$) receives from $r$ an $\langle \ID_u, \PH_u, \DECIDE\rangle$ message, it must decide whether to end its candidacy (if it is still a candidate) or that of the other candidate. To do so, it compares its own current position with that of the contender $u$. 

\begin{itemize}
\item 
If $v$ is in candidate state $\nonelectedSTATE$ or $v\ll u$, then $v$ returns the message $\langle \ID_u,\PH_u,\WINS\rangle \circ \langle \ID_v,\PH_v, \LOSES \rangle$. 
\item
Otherwise ($v$ is still in candidate state $\candidateSTATE$ and is ahead of $u$), it returns the message $\langle \ID_u, \PH_u, \LOSES\rangle \circ \langle \ID_v,\PH_v,\WINS\rangle$. 
\end{itemize}
Pseudocode for the chosen's actions on receiving a $\DECIDE$ message from a referee is given in Procedure~\ref{alg:decide}.

\alglanguage{pseudocode}
\begin{algorithm}
	\caption{Procedure Decide, run by a node $v$ upon receiving the message $\langle \ID_u,\PH_u, \DECIDE\rangle $ from node $r$.}
	\label{alg:decide}
	
	\begin{algorithmic}[1]

	\If{$\CANDSTATE = \nonelectedSTATE$}
		\State Send the message  $\langle \ID_u,\PH_u, \WINS\rangle \circ \langle \ID_v, \PH_v, \LOSES\rangle$ to $r$
	\ElsIf{$\CANDSTATE = \candidateSTATE$}
		\If{$u\gg v$}
			\State $\CANDSTATE \gets \nonelectedSTATE$
			\State Send the message  $\langle \ID_u,\PH_u, \WINS\rangle \circ \langle \ID_v, \PH_v, \LOSES\rangle$ to $r$
		\Else
			\State Send the message $\langle \ID_u,\PH_u, \LOSES\rangle \circ \langle \ID_v, \PH_v, \WINS\rangle$ to $r$
		\EndIf

	\EndIf

	\end{algorithmic}
\end{algorithm}

Once the chosen's reply message is received by $r$, it is processed as follows.
\begin{itemize}
\item
If $v$ replied that the contender $u$ has won the dispute, $r$ sends an approval message to the current contender (which is $u$ if the referee state is $C2$, and another candidate if the referee state is $C3$), makes it the chosen, and switches to referee state $C1$. 
\item
If $v$ replied that it has progressed beyond the contender $u$, so $u$ has lost the dispute and must retire, then there are two sub-cases to consider. If the referee state is $C2$, $r$ sends a decline message  to the contender $u$, and switches to state $C1$.
Now suppose the referee state is $C3$ and the current contender is some candidate $w$.
If $v$'s current position is such that $v\gg w$, then $r$ sends a decline message to $w$ and switches to referee state $C1$. Otherwise ($w\gg v$), a new dispute is required, this time between $v$ and $w$.
\end{itemize}
The referee's actions on receiving a reply from the chosen about an ongoing dispute are given in Procedure~\ref{alg:referee2}.

\alglanguage{pseudocode}
\begin{algorithm}
	\caption{Procedure Referee\_Dispute\_Response, run by a referee $r$ on receiving a reply to a message $\langle\ID_u,\PH_u,\DECIDE\rangle$ from the chosen $v$. This message only arrives while the referee is in state $C2$ or $C3$.}
	\label{alg:referee2}
	
\begin{algorithmic}[1]
\State Let $w$ be the $\Contender$ ~~(* which may be different from $u$ *)
\If{reply is $\langle \ID_u, \PH_u, \WINS\rangle \circ \langle \ID_v, \PH_v, \LOSES\rangle$ ~(* $u$ won the dispute *)}
    \State Send message  $\langle \ID_w, \PH_w, \APPROVED\rangle $ to node $w$
    \State Set $\Chosen\gets \Contender$; $\Contender\gets \bot$
    \State Set $\REFSTATE \gets C1$
\Else ~(* The chosen $v$ won the dispute over $u$ *) 
    \If{$\REFSTATE=C2$ OR ($\REFSTATE=C3$ AND $v\gg w$)}
        ~~ (* $v$ remains the chosen *)
        \State Send message $\langle \ID_w, \PH_w, \DECLINED\rangle$ to node $w$
        \State $\Contender\gets \bot$; Set $\REFSTATE \gets C1$
    \Else ~(* $\REFSTATE=C3$ but $w\gg v$, so a new dispute is required *)
        \State Send message $\langle\ID_w,\PH_w,\DECIDE\rangle$ to the chosen $v$
        \State Set $\REFSTATE \gets C2$
    \EndIf
\EndIf
\end{algorithmic}
\end{algorithm}

\smallskip 
\noindent \textbf{Analysis of the algorithm.} We establish the following theorem.

\begin{theorem}\label{the:logn-rounds-asynch-rand}
Consider a complete anonymous network $G$ of $n$ nodes in the $\mathcal{CONGEST}$ model.
Assume communication is asynchronous with adversarial wakeup. Then there is a randomized algorithm to solve leader election with high probability in $O(\log^2 n)$ time with high probability
using $O(n)$ messages with high probability. If each node has a unique identifier, then the algorithm always succeeds. All nodes terminate at the end of the algorithm.
\end{theorem}

We show three properties:  exactly one leader is chosen w.h.p.\ for anonymous networks and always when nodes have unique identifiers and all nodes subsequently terminate, the time complexity is $O(\log^2 n)$ w.h.p., and the message complexity is $O(n)$ w.h.p.

Before we continue with the proof, we make an important observation.
\begin{observation}\label{obs:node-start-phase-end-phase}
For each node $u$ that participates in the algorithm, if $u$ begins phase $i$, then $u$ will 
finish phase $i$.
\end{observation}

Observation~\ref{obs:node-start-phase-end-phase} is used implicitly in the remaining proof whenever a node is mentioned to finish some phase and possibly perform some calculation as a result of completing that phase.

\begin{lemma}\label{lem:exactly-one-leader}
By the end of the algorithm, exactly one node is in candidate state $\electedSTATE$ w.h.p. for anonymous networks and always when nodes have unique identifiers and the remaining nodes are in candidate state $\nonelectedSTATE$. Furthermore, all nodes eventually terminate.
\end{lemma}

\begin{proof}
Let us consider a network where each node has a unique identifier. We show that exactly one leader is always chosen by arguing that exactly one $\ID$ is chosen as the leader. It is clear to see that if we then consider an anonymous network where each node chooses its $\ID$ from $[1,n^4]$ and nodes do not have access to unique identifiers to break ties, the $\ID$ chosen as the leader will belong to not more than one node w.h.p.

We prove by induction that at least one node is a candidate at the end of phase $i$ for $1 \leq i \leq \numphaserand$. 
For the induction basis, consider, for the sake of the proof, that when a node wakes up, it is in phase 0, but it immediately moves to phase 1. Clearly, this does not change the outcome of the algorithm. Now the lemma holds for phase zero since at least one node is woken up by the adversary initially, and the induction is proven with phase zero as base case.

Assume that the claim holds true up to the end of some phase $k < \numphaserand$. Let $u$ be the candidate with largest \ID\ in phase $k+1$. If $u$ is not retired by any other node in phase $k+1$, then it will be in candidate state $\candidateSTATE$ at the end of phase $k+1$. Hence the claim holds for $k+1$. Otherwise, if $u$ is retired in phase $k+1$, then necessarily either $u$ came into contact with (i) some candidate $u'$ in a higher phase $\PH_{u'}> k+1$ or a referee to such a candidate, (ii) a candidate $u''$ with higher rank $\ID_{u''}>\ID_u$ in the same phase $k+1$ or a referee to such a candidate, or (iii) a candidate $v$ that completed all phases and is still in candidate state $\candidateSTATE$ or a message from such a candidate $v$. In all cases, the induction claim still holds true for $k+1$. 

We subsequently show that exactly one node will change its candidate state to $\electedSTATE$. By the previous claim, at least one candidate completes phase $\numphaserand-1$. Let $u_{max}$ be the node with the largest \ID\ that is still a candidate at the beginning of phase $\numphaserand$. In phase $\numphaserand$, $u_{max}$ will approach all other nodes, receive an approval from each of them, and subsequently change its candidate state to $\electedSTATE$. Any other candidate $v$ will approach $u_{max}$, receive a $\DECLINED$ message from it, and thus change its candidate state to $\nonelectedSTATE$. Thus exactly one node will change its candidate state to $\electedSTATE$ and broadcast a $\langle \cdot, 0, \LEADER\rangle $ message.

Finally, we argue that all nodes terminate. It is clear that if exactly one node changes its candidate state to $\electedSTATE$ and broadcasts a $\langle \cdot, 0, \LEADER\}$ message, then that node terminates the algorithm. All other nodes change their candidate state to $\nonelectedSTATE$ and terminate the algorithm upon receiving this message. Thus, the proof is complete.
\end{proof}
%
We now prove the time complexity bound. First, we prove the following useful lemma about the amount of time one phase for a candidate takes. 

\begin{lemma}\label{lem:time-bound-phase-O1}
If a node $u$ starts some phase $i$ as a candidate, then it exits the phase (either by increasing its phase to $i+1$ or by retiring) within $O(1)$ time.
\end{lemma}

\begin{proof}
For any candidate $u$, for each phase $i$, the largest delay until that phase ends occurs when candidate $u$ sends a message to a referee $r$, which in turn sends a message to its current chosen $v$. Subsequently, $v$ sends a reply to $r$, which in turn sends a reply to $u$. If the congestion on both these edges is $O(1)$ messages in the time interval from start to end of phase $i$, then the total time duration of the phase is $O(1)$. Let $e_{u,r}$ and $e_{r,v}$ denote edges between $u$ and $r$ and between $r$ and $v$ resp. We show at most $O(1)$ messages need to be transmitted on both $e_{u,r}$ and $e_{r,v}$ during phase $i$, thus allowing each phase to take $O(1)$ time.

Consider edge $e_{u,r}$. In phase $i$, $u$ can send an invite to $r$ and receive a reply to the invite from $r$. Furthermore, $u$ can also be the chosen for $r$ and receive a $\DECIDE$ message from $r$ and subsequently send back a reply.\footnote{Recall that once some node $u$, in phase $i$, is the chosen of some node $r$ and receives a $\DECIDE$ message, either $u$ is retired or $r$ learns that $u$ is in phase $i$ and will not send any further messages from other candidates in phase $<i$. If another $\DECIDE$ message is subsequently sent to $u$ from $r$, then either $u$ wins (implying it is in a higher phase), or else $u$ is retired and no further messages are sent along the edge.} Additionally, $r$ may also be a candidate with $u$ as its referee or $r$ may be the chosen for $u$, resulting in at most a doubling of messages over the edge. Thus at most $O(1)$ messages are sent across that edge for phase $i$ of $u$. Now consider edge $e_{r,v}$. A similar analysis as above renders $O(1)$ an upper bound on the messages on the edge. Thus, we see that each phase takes $O(1)$ time units.
\end{proof}

In order to bound the running time, we also make use of the following useful property from~\cite{P20}, denoted in the reference as Problem C.2 in Appendix C.6. We slightly modify the statement to suit our needs.

\begin{lemma}{(Problem C.2 in~\cite{P20})}\label{lem:sum-exp-values-bound}
Let $a_1, a_2, \ldots, a_n$ be a sequence of $n$ values. Each value $a_i$ is independently and randomly chosen from a fixed distribution $\mathcal{D}$. Let $m_i = \max \lbrace a_1, a_2, \ldots, a_i \rbrace$, i.e., the maximum of the first $i$ values. Let random variable $Y$ denote the number of times the maximum value is updated, i.e., the number of times $m_i \neq m_{i+1}$. Then $E[Y] = O(\log n)$.
\end{lemma}

\begin{lemma}\label{lem:running-time-rand}
The running time of the algorithm is $O(\log^2 n)$  with high probability. 
\end{lemma}

\begin{proof}
For any given execution of the algorithm, eventually one candidate, $C_{\ell}$, woken up by the adversary, becomes the leader. We show two properties for any such $C_{\ell}$. First, once $C_{\ell}$ is woken up by the adversary, the algorithm takes an additional at most $O(\numphaserand)$ time until $C_{\ell}$ broadcasts that it is the leader and the algorithm terminates for every node. Second, we bound the time from when the first node is woken up by the adversary to when candidate $C_{\ell}$ is woken up  as $O(\numphaserand \cdot \log n)$ with high probability.

Once $C_{\ell}$ is woken up, it must complete $\numphaserand$ phases before it broadcasts that it is the leader. By Lemma~\ref{lem:time-bound-phase-O1}, each of its phases takes $O(1)$ time to complete. Thus, once $C_{\ell}$ is woken up by the adversary, it takes at most $O(\numphaserand)$ time until all nodes terminate.

Now, we bound the time until $C_{\ell}$ is woken up. For this proof, we say that two candidates $u$ and $v$ \textit{encounter} each other at time $t$ when a referee (which could be either $u$ or $v$ itself) is made aware of both $u$ and $v$ for the first time. If, as a result of the encounter, $v$ is retired, then we call $u$ the \textit{winner} of the encounter and $v$ the \textit{loser}. Note the following observation.

\begin{observation}\label{obs:time-complexity-rand}
Consider some execution of the algorithm where a candidate $u$ while in phase $k$ retires candidate $w$. Subsequently, $u$ reaches phase $i$ and is retired by another candidate $v$ in phase $j$. When $u$ and $v$ encounter each other, then either $j > i$ or $v$ has a larger $\ID$ than $u$ and is also in phase $i$. Furthermore, from the time $u$ retired $w$, at most $O(i-k)$ time units have passed until the encounter between $u$ and $v$ (by Lemma~\ref{lem:time-bound-phase-O1}) and an additional $O(1)$ time units pass until $u$ retires.
\end{observation}

Observation~\ref{obs:time-complexity-rand} has the following implications. Either the winner of an encounter has a larger $\ID$ than the loser or else the phase number of the winner is larger than that of the loser when they encounter one another. Furthermore, the amount of time that passes between two encounters involving a common candidate can be bounded by the difference in phases of that candidate plus $O(1)$ time units. 

We now construct a new graph where the nodes are the subset of the nodes that are woken up by the adversary (and which become candidates). There is an edge from candidate $C_i$ to $C_j$ if candidate $C_i$ encounters and subsequently is retired by $C_j$. Notice that this graph is a directed acyclic graph since once a candidate $C_i$ is retired by some other node $C_j$, $C_i$ cannot go on to retire $C_j$ or any candidate that retires $C_j$. Notice also that a candidate may encounter multiple other candidates, resulting in nodes with an in degree or out degree (or both) that is $>1$. Finally, note that candidates that do not retire any other candidates are sources and $C_{\ell}$ is a sink. Consider a path in the graph from a source to $C_{\ell}$ such that the node first woken up by the adversary is one of the nodes in this path. If there are multiple such paths, choose one arbitrarily. Label the nodes from the source to the node just before $C_{\ell}$ as $C_1, C_2, \ldots, C_w$. Let $C_1$ be in some phase $p$ at the time of its encounter with $C_2$. Let $R$ be the time that pass from the first encounter until the final encounter in the chain. Then, the total time from when the first node was woken up until $C_{\ell}$ was woken up is upper bounded by $O(p + R) = O(\numphaserand + R)$. We now bound the value of $R$.

With each of the candidates $C_i$, $1 \leq i \leq w$, associate a bit $b_i$ which indicates how $C_i$ was retired. Specifically, $b_i$ is set to $0$ if $C_{i+1}$ was in a higher phase than $C_i$ at the time of the encounter.\footnote{Note that we are referring to the actual phases $C_i$ and $C_{i+1}$ are in at this time, not necessarily the phase numbers stored in the referee when it was first aware of both of them, as one of these values may be outdated. Also note that if $i=w$, then $C_{i+1}$ refers to $C_{\ell}$.} Bit $b_i$ is set to $1$ if $C_{i+1}$ has a higher $\ID$ than $C_i$ and is in the same phase as it at the time of the encounter. By Observation~\ref{obs:time-complexity-rand} and its implications, we see that every bit $b_i$, $1 \leq i \leq w$ is set to $0$ or $1$.

Consider the bit string $B = b_1 b_2 \ldots b_w$. We will show that $R$ is upper bounded $O(\numphaserand + |B|)$. In order to bound the size of $B$, we first bound the number of $0$'s that can be present in $B$, and then we bound the number of $1$'s that can be present between any two $0$'s in $B$.

Denote by $P_i$ the phase of $C_i$ when it was retired. Observe that, regardless of whether $b_i = 0$ or $1$, the phase $P_{i+1}$ of the subsequent candidate $C_{i+1}$ can never be less than that of the candidate $C_i$ it just retired. Furthermore, when $b_i = 0$, $P_{i+1} > P_i$. Thus, there can be at most $\numphaserand$ bits set to $0$ in $B$. We now bound the number of bits set to $1$ between any two bits that are set to $0$.

Between any two bits set to $0$, at most $n-1$ nodes can be woken up by the adversary to trigger an encounter leading to a bit being set to $1$. An encounter where a bit is set to $1$ involves the $\ID$ of the awakened node being higher than that of the currently considered candidate, resulting in the higher $\ID$ node becoming the currently considered candidate. By Lemma~\ref{lem:sum-exp-values-bound}, we see that on expectation, $O(\log n)$ such encounters can thus be triggered. Applying a simple Chernoff bound, we see that $O(\log n)$ is in fact a high probability upper bound on the number of such encounters, and by extension the number of bits set to $1$ between any two bits set to $0$. Since there are at most $\numphaserand$ bits set to $0$, $|B| = O(\numphaserand \cdot \log n)$.

Each encounter of  two candidates contributes $O(1)$ time towards the running time. Furthermore, we must account for the time between encounters as well. Between any two encounters, the phase of a candidate may increase. We have already seen that if candidate $C_i$ is in phase $P_i$, then for subsequent candidates $C_j$ where $j>i$, $P_j \geq P_i$. From Lemma~\ref{lem:time-bound-phase-O1}, we see that for any candidate a phase takes $O(1)$ time. Therefore, the total number of time units between all encounters is at most $O(\numphaserand)$ time units. Thus, $R$ can be upper bounded by $O(\numphaserand + |B|) = O(\numphaserand \cdot \log n)$ and the total running time is $O(\numphaserand + R) = O(\log^2 n)$ time with high probability.
\end{proof}


We now prove that the message complexity of the randomized algorithm is $O(n)$ with high probability. We first show that the number of candidates that can participate in every phase is reduced by a factor of four from one phase to the next up to phase $\rho = \numphaserand - \lceil \log \log n \rceil - 5$ with high probability. 

\begin{lemma}\label{lem:limit-num-candidates-per-phase-rand}
The total number of candidates that participate in phase $i$ of the algorithm is at most $\lceil n/4^{i-1} \rceil$ with high probability for $1 \leq i \leq  \rho$.
\end{lemma}

\begin{proof}
We prove the claim by induction. Initially, even if the adversary wakes up all nodes, at most $n$ candidates participate in phase $1$. Assume that the claim holds true until some phase $k < \rho$. We prove that the claim holds in phase $k+1$ when the number of candidates participating in phase $k$ is upper bounded by $\lceil n/4^{k-1} \rceil$.

Notice that if $\leq \lceil n/4^k \rceil$ candidates participate in phase $k$, then the claim holds immediately for phase $k+1$. Now, let us assume that the number of candidates participating in phase $k$ lies in $(\lceil n/4^k \rceil, \lceil n/4^{k-1} \rceil]$. 

Organize the candidates in increasing order of their $\ID$s. Group the top one-sixteenth of these candidates into the set $\textit{top}$ and group the remaining candidates into the set $\textit{bottom}$. We ignore the retirement of candidates from $\textit{top}$ in phase $k$ by assuming that none of them are retired in this phase
and show that a sufficient number of candidates from $\textit{bottom}$ are retired for the claim to hold in phase $k+1$.\footnote{Notice that the claim is an upper bound on the number of candidates in each phase. By showing that the claim holds when none of the candidates from $\textit{top}$ is retired, it is easy to see that the claim also holds when at least one candidate from $\textit{top}$ is retired.} In this phase, there are at least $1/16 \cdot \lceil n/ 4^k \rceil$ candidates in $\textit{top}$, each of which makes $10 \cdot 2^k$ requests.\footnote{Notice that for all phases $i$ that the claim applies to, $10 \cdot 2^i \leq \lceil \sqrt{4n \log n} \rceil$. Thus, any candidate in one of these phases will make $10 \cdot 2^i$ requests.}
Thus, altogether, candidates from $\textit{top}$ make at least 
$m_0 = (\lceil n/4^k \rceil /16) \cdot 10 \cdot 2^k \geq 5 n / 2^{k+3}$ requests uniformly at random.\footnote{Note that the requests by a candidate in a single phase are made without repetition, i.e. a candidate does not send more than one request along the same edge in a single phase. During the subsequent analysis, we consider them to be made with repetition so as to simplify the analysis. This assumption only decreases the probability of a candidate from $\textit{bottom}$ inviting a referee that was approached by one of the $\textit{top}$ candidates and thus this simplifying assumption is acceptable.}

Consider a candidate $u$ from $\textit{bottom}$. Let $r$ be some specific referee approached by $u$. 

Call $r$ {\em top-free} if none of the candidates from $\textit{top}$ has approached $r$ and $r$ itself is not in $\textit{top}$. Then the probability that $r$ is top-free is at most $p_0 = (1-1/n)^{m_0} \cdot 15/16 \leq (1-1/n)^{m_0}$. Note that $u$ will be retired if even one of its requests is to a non-top-free referee.
Therefore, the probability that $u$
is not retired in this phase
is at most $p_0^{10\cdot 2^k} \le (1-1/n)^{(5 n/2^{k+3}) \cdot 10 \cdot 2^k} \leq e^{-25/4}$.
Thus, denoting the number of candidates from $\textit{bottom}$ that are not retired by $NR_B$, we have 
\begin{align*}
\Exp{NR_B} &\leq (15/16) \cdot \lceil n/ 4^{k-1} \rceil \cdot e^{-25/4} \\
&\leq (15/2) e^{-25/4}\cdot n/4^{k}  \\
&\leq (1/16) \cdot (n/ 4^k).
\end{align*}
As $NR_B$ is the sum of independent Bernoulli trials, we can apply a Chernoff bound (second bound of Theorem~4.4 in \cite{MU17}). When $k < \rho$, the probability that $NR_B$ exceeds twice the expected value is upper bounded by $1/n^3$.
Thus the total number of candidates that can participate in phase $k+1$ with high probability is at most $(1/16) \cdot \lceil n/4^{k-1} \rceil + (2/16) \cdot (n/ 4^k) \leq \lceil n/ 4^k \rceil$.

Note that for each phase, the upper bound on the number of candidates who are not retired holds with probability $1 - 1/n^3$, assuming that the bound held in the previous phase. Applying a union bound over all phases of the induction, we see that this upper bound holds for phase $\rho$ with probability $1 - O(1/n^2)$, i.e.\ w.h.p., and for each previous phase with a larger probability. Thus the claim is true for all $1 \leq i \leq \rho$ with high probability.
\end{proof}

We next show that in the final phase of the algorithm, only a single candidate will not be retired with high probability.

\begin{sloppypar}
\begin{lemma}\label{lem:one-candidate-survives-final-phase-rand}
At the end of phase $\numphaserand-1$, exactly one node will be in candidate state $\candidateSTATE$ with high probability.
\end{lemma}
\end{sloppypar}

\begin{proof}
Consider the candidate with the largest \ID, $u_{max}$, in phase $\numphaserand-1$. We show that for each candidate $v \neq u_{max}$ in this phase, the intersection of referees approached by $v$ and $u_{max}$ is non-zero with high probability. Thus $u_{max}$ will retire each such candidate $v$ with high probability. Notice that $10 \cdot 2^{\numphaserand} \geq \lceil \sqrt{4n \log n} \rceil$ and so each candidate in this phase makes exactly $\lceil  \sqrt{4n \log n} \rceil$ requests.
Then
$$\Prob{v~\mbox{is not retired by}~u_{max}} 
~=~ (1 - 1/n)^{\lceil  \sqrt{4n \log n} \rceil\cdot\lceil  \sqrt{4n \log n} \rceil} ~\leq~ 1/n^4.$$

By Lemma~\ref{lem:limit-num-candidates-per-phase-rand}, there are at most $\lceil 2^{13} \log n \rceil$ candidates participating in phase $\rho$. As the number of candidates participating in subsequent phases cannot increase beyond this value, it acts as an upper bound for the number of candidates participating in phase $\numphaserand-1$.

We see that the probability that any of these candidates is not retired is at most $1/n^3$ through the use of a union bound. Thus, w.h.p. $u_{max}$ retires every other candidate in this phase. Thus only $u_{max}$ ends the phase in candidate state $\candidateSTATE$ w.h.p.
\end{proof}

\begin{lemma}
The message complexity of the algorithm is $O(n)$ messages with high probability.
\end{lemma}

\begin{proof}
\begin{sloppypar}
Let $\psi = \lfloor \log (1/10 \cdot \lceil \sqrt{4n \log n} \rceil) \rfloor$. 
We start by analysing the message complexity that can be ``charged'' to candidates in phases 1 to $\psi$.
In each phase $i \leq \psi $, each candidate $u$ generates and receives up to $2 \cdot 10\cdot 2^i$ messages from its referees. Furthermore, for each of its referees in a given phase $i$ and from all previous phases, a candidate may receive a message to decide candidacy, resulting in at most 
$\sum_{j=1}^i \sum_{k=1}^{j} 2 \cdot 10 \cdot 2^k \leq \sum_{j=1}^i 2 \cdot 10 \cdot 2^{j+1} = 5 \cdot 2^{i+4}$ 
additional messages received and generated by the candidate in this phase. 
(For a candidate $u$ and a referee $r$ chosen by $u$ in phase $i$ or a previous one, only one $\DECIDE$ message will be sent from $r$ to to $u$.) 

In each phase $1 \leq i \leq \rho$, there are at most $\lceil n/4^{i-1} \rceil$ candidates w.h.p. by  Lemma~\ref{lem:limit-num-candidates-per-phase-rand}. As calculated earlier, we have $5 \cdot 2^{i+2} + 5\cdot 2^{i+4}=25 \cdot 2^{i+2}$ messages per candidate in each phase $i$, resulting in a total of $O(n)$ messages w.h.p. generated across all candidates in all phases.

In phases $\rho + 1 \leq i \leq \psi$, there are at most $\lceil 2^{13} \log n \rceil$ candidates in each phase with high probability by Lemma~\ref{lem:limit-num-candidates-per-phase-rand}, each making $10 \cdot 2^i$ requests in that phase. Thus, the total number of messages over these phases is $\sum_{i =  \rho + 1}^{\psi} \lceil 2^{13} \log n \rceil 25 \cdot 2^{i+2} = O(\sqrt{n} \log^{3/2} n)$ w.h.p. 

Now we analyze the message complexity attributed to candidates in phases $\psi$ and above.
In each phase $i > \psi$, each candidate $u$ generates and receives $2 \cdot \lceil \sqrt{4n \log n} \rceil$ messages from its referees. Furthermore, for each of its referees in a given phase $i$ and from all previous phases, $u$ may receive a message to decide candidacy, resulting in at most 
$2\lceil \sqrt{4n \log n} \rceil (i-\psi)$ additional messages due to $\DECIDE$ messages from referees of phases $j>\psi$, plus
$5\cdot 2^{\psi+4}$ additional messages due to $\DECIDE$ messages from referees of phases $j\le\psi$. Thus, at most $2 \cdot (i + 5 - \psi) \cdot\lceil \sqrt{4n \log n} \rceil$  messages are received and generated by $u$ in this phase.

In phases $i > \psi$, there are most $\lceil 2^{13} \log n \rceil$ candidates in each phase w.h.p. by Lemma~\ref{lem:limit-num-candidates-per-phase-rand}. Thus, the total number of messages across all these phases is at most $\sum_{i =  \psi + 1}^{\numphaserand} \lceil 2^{13} \log n \rceil \cdot 2 \cdot (i + 5 - \psi) \cdot\lceil \sqrt{4n \log n} \rceil = O(\sqrt{n} \log^{5/2} n)$ w.h.p.

By Lemma~\ref{lem:one-candidate-survives-final-phase-rand}, exactly one candidate will complete all the first $\numphaserand-1$ phases and remain in candidate state $\candidateSTATE$ w.h.p. This candidate will generate $O(n)$ messages in phase $\numphaserand$.

\end{sloppypar}

Thus, over all phases of the algorithm, there are totally $O(n)$ messages w.h.p.
\end{proof}

\section{A (Tightly) Singularly Optimal Synchronous Algorithm}
\label{sec:adv-wakeup-synch-alg}
In this section, we present a message and time optimal algorithm for the synchronous setting, i.e., an algorithm that takes $O(1)$ time and $O(n)$ messages with high probability. Note that these upper bounds are tight (hence the term ``tightly'' in the section title). Recall that we are still dealing with adversarial wake up, so a node may still not be woken up by the adversary, or may wake up late.

Like in the asynchronous algorithm of Section \ref{sec:adv-wakeup-asynch}, when a node wakes up, it chooses a $\ID$ from $[1,n^4]$ uniformly at random.
Also, nodes participate in the algorithm as candidates or referees. However, here there are two types of candidates, {\em silent} and {\em active}. Specifically, a node $v$ that is woken up spontaneously (by the adversary) becomes a silent candidate.
A silent candidate turns into an active candidate following a successful coin toss, as described later.
A node $v$ can also be woken up by another node $w$ who is an active candidate and who approaches $v$ with an approval request. In that case, $v$ becomes a referee and will not propose its candidacy. Finally, a node $v$ can also be woken up by a $\WINNER(\ID)$ message announcing the rank of a winner, in which case it does nothing.

If a silent candidate $v$ is approached at any time by an active candidate $w$ appointing it as its referee and requesting its approval (or by a $\WINNER(\ID)$ message), then $v$ immediately retires (and acts as a referee if requested).

A node $v$ that is woken up at time $t$ and becomes a silent candidate, makes three attempts to transform into an active candidates. 
Each attempt corresponds to a biased coin toss. If the coin comes up heads, the candidate transforms from a silent candidate into an active one.
At time $t$, $v$ selects itself randomly as an active candidate with probability $n^{-2/3}$. If this attempt fails, then $v$ remains silent for 3 time units and (assuming it was not retired in the meantime by some active candidate) tries again at time $t+3$, this time with probability $n^{-1/3}$.
If this attempt fails as well, then $v$ remains silent for 3 additional time units and (assuming it was not retired) tries again at time $t+6$, this time with probability $1$.

Once a silent candidate $v$ becomes an active candidate, it actively attempts to become the leader. To do so, $v$ chooses u.a.r.\
(uniformly at random)
$\lceil2\sqrt{n} \log n\rceil$ neighbors as referees, and sends them an approval request containing its rank, $\ID_v$. 
For simplicity, $v$ always appoints itself as its own referee in addition to the previously chosen ones.

Note that there may be other active candidates trying to get elected at the same time. Among them, the one with the largest $\ID$ will win the election.\footnote{If nodes have unique identifiers, then in case two active candidates have the same $\ID$, their unique identifiers will be used to break ties. Furthermore, their unique identifiers will be appended to any messages involving their $\ID$ to avoid ambiguity.}
Each referee $w$, approached by a group $C_w$ of active candidates ($|C_w|\ge 1$) at some time $t$, 
selects the largest $\ID_v$ among all $v\in C_w$, and sends a message back to each $v\in C_w$ with this $\ID$. (Note that $w$ itself may be one of the competing active candidates.)

If an active candidate receives its own $\ID$ back from all its referees, then $v$ declares itself a {\em winner} by broadcasting $\WINNER(\ID_v)$.
(Note that $v$ is not guaranteed to be the elected leader, since there may be more than one winner, but w.h.p. there will be a single winner.) 
Otherwise (i.e., if $v$ hears $\ID_{v'}$ from some referee $w$, where $\ID_{v'}>\ID_v$), $v$ retires.

At any time, a node $v'$ receiving one or more $\WINNER(\ID_v)$ messages finds the largest $\ID$ among these winners and sets $leader \gets \ID$. This holds even if $v'$ itself is currently an active candidate and $\ID_{v'}>\ID$. (This can happen if $v'$ became active at a later time than the earliest active candidates, and still did not send its own $\WINNER$ message.)
We establish the following theorem.

\begin{theorem}\label{the:synch-alg-adv-wakeup}
Consider a complete anonymous network $G$ of $n$ nodes in the $\mathcal{CONGEST}$ model, with synchronous communication and adversarial wakeup. There is
a randomized algorithm to solve leader election w.h.p. using $O(n)$ messages w.h.p. in $9$ rounds (deterministically). (If nodes have unique identifiers, then the algorithm always elects a leader.)
\end{theorem}

\begin{proof}
We first argue about the time and correctness. We then prove that the bound on message complexity holds w.h.p.

\begin{lemma}
The algorithm elects exactly one leader w.h.p. for an anonymous network and always elects exactly one leader when nodes have unique identifiers. Furthermore, the algorithm terminates within at most 9 time units from the time the first node is woken up.
\end{lemma}
\begin{proof}
Suppose the adversary wakes up the first node at time $t_0$. Note that the adversary may wake up more than one node at time $t_0$, as well as in any later time $t>t_0$.
However, no messages are sent until the first time $\hht$ in which some silent candidate succeeds in becoming active. 

Observe that $\hht\le t_0+6$, since (i) from time $t_0$ and after there is at least one silent candidate, (ii) a silent candidate that is not retired will necessarily become active at most 6 time units after waking up, by which time its success probability is increased to 1, and (iii) a silent candidate can only be made to retire (directly or indirectly) on some time $t$ by an active candidate who necessarily became active earlier than $t$. 

Observe that once at least one active candidate exists, a leader will be elected for sure within at most 3 additional time units. In particular, denoting by $Z$ the set of silent candidates that became active at time $\hht$, and letting $z$ be the highest $\ID$ node among them, $z$ will surely announce itself a leader. 

Note that messages will be sent out only at times $\hht$, $\hht+1$ and $\hht+2$. Specifically, at time $\hht$, all nodes in $Z$ will send their requests to their referees, the referee replies will be sent in the following time, and at time $\hht+2$ each elected leader will send out its announcement. 

Note that with high probability, $z$ will be the only winner. To see this, observe that each active candidate chooses some $\lceil 2\sqrt{n}\log n\rceil$ referees u.a.r. For each other active candidate $v_i\in Z$, define a bad event $E_i$ occurring when there is no overlap between $v_i$'s and $z$'s referees. If $E_i$ occurs, $v_i$ could possibly become a winner. However, 
$$Prob[E_i]
~=~ (1 - 2\sqrt{n} \log n/n)^{2\sqrt{n} \log n} 
~=~ (1 - 2\log n/\sqrt{n})^{2\sqrt{n} \log n} ~\leq~ e^{- 2\log n} ~=~ 1/n^2.$$
As there are at most $n-1$ active candidates in $Z$, by the union bound, the probability that at least one bad event occurs is at most $(n - 1)/n^2 \leq 1/n$. Thus, w.h.p.\ no bad event occurs and only $z$ becomes a winner.

In the unlikely case when multiple candidates become winners and broadcast their rank at time $\hht+2$, if nodes have unique identifiers, then exactly one of the winners will become the elected leader of everyone in the network, since node $\ID$s would then be distinct.

Note that some silent candidate $q$ may become active at time $\hht+1$ or $\hht+2$ and approach its referees. Such a candidate will not win over $z$ even if $\ID_q>\ID_z$, since it will get the message $\WINNER(\ID_z)$ one or two time units before sending out its own announcement.  
\end{proof}

\begin{lemma}
The message complexity of the algorithm is $O(n)$ with high probability.
\end{lemma}

\begin{proof}
As mentioned earlier, throughout the execution, no messages are sent except at times $\hht$, $\hht+1$ and $\hht+2$. 

Each winner generates $O(n)$ $\WINNER$ messages in its final broadcast. 
As argued earlier, with high probability, at most one node will send $\WINNER$ messages, so the contribution of such messages to the total message complexity is $O(n)$. 

Each of the candidates in $Z$ is responsible for $O(\sqrt{n}\log n)$ ``election messages'' between itself and its referees. 
In addition, such messages may be sent by candidates who become active at time $\hht+1$ or $\hht+2$ (although these candidates will not become winners).
Denote the sets of these late-coming active candidates by $Z'$ and $Z''$ respectively.
Hence to bound the number of messages, it suffices to bound $|Z| + |Z'| + |Z''|$. 

Let us first bound $|Z|$.
Partition the set $Z$ into three subsets, $Z=\bigcup_{\ell=1}^3 Z_\ell$, where 
$Z_1$ consists of nodes that woke up at time $\hht$ and succeeded in becoming active candidates on their first attempt, 
$Z_2$ consists of nodes that woke up at time $\hht-3$ and succeeded in becoming active candidates at time $\hht$ on their second attempt, and
$Z_3$ consists of nodes that woke up at time $\hht-6$ and succeeded on their third attempt. 
For $\ell=1,2,3$, denote by $W_\ell$ the set of silent candidates awoken by the adversary at time $\hht-3(\ell-1)$, and let $n_\ell=|W_\ell|$.
Note that $\Exp{|Z_\ell|} = n_\ell \cdot n^{\ell/3-1}$.

Consider first the active candidates of $Z_1$ (i.e., that woke up at time $\hht$).
Since $n_1\le n$, their expected number is
$$\Exp{|Z_1|} ~=~ n_1 \cdot n^{-2/3} ~\le~ n^{1/3}.$$
Therefore, with high probability, the actual number of active candidates in $Z_1$
satisfies $|Z_1| ~\le~ 6 n^{1/3}$. (To prove this, consider the third inequality from Theorem 4.4 in~\cite{MU17}, i.e., $Prob(X\ge R) < 2^{-R}$ for $R\ge 6\Exp{X}$, where $X$ is the sum of independent Poisson trials. Setting $R = 6 n^{1/3}$, we get the desired high probability bound.)

Next, consider candidates from $Z_2$.
Note that with high probability,
\begin{equation}
\label{eq: f-1}
n_2 \le n^{7/9}.
\end{equation}
To see why this holds, suppose towards contradiction that $n_2 > n^{7/9}$.
At time $\hht-3$, each of the $n_2$ silent candidates in $W_2$ attempted to become active with probability $n^{-2/3}$. Therefore, the expected number of silent candidates from $W_2$ that should have become active at time $\hht-3$ was
$$n_2 \cdot n^{-2/3} ~>~ n^{7/9} \cdot n^{-2/3}
~=~ n^{1/9}.$$
This implies that with high probability, there should have been at least one actual active candidate from $W_2$ at time $\hht-3$, contradicting the definition of $\hht$.

By Eq. \eqref{eq: f-1}, the expected number of active candidates in $Z_2$ is thus
$$\Exp{|Z_2|} ~=~ n_2 \cdot n^{-1/3} ~\le~ n^{7/9} \cdot n^{-1/3} ~\le~ n^{4/9}.$$
Therefore, with high probability, 
$|Z_2| ~\le~ 6 n^{4/9}$.

Finally, consider the active candidates in $Z_3$.
By arguments similar to those used to show Eq. \eqref{eq: f-1}, we have that with high probability, $n_3 \le n^{4/9}$.
It follows that the expected number of active candidates in $Z_3$ is
$$\Exp{|Z_3|} ~=~ n_3\cdot 1 ~\le~ n^{4/9}.$$
Therefore, with high probability, $|Z_3| ~\le~ 6 n^{4/9}$.

In summary, $|Z|=|Z_1|+|Z_2|+|Z_3| \le 6 n^{1/3} + 6 n^{4/9} + 6 n^{4/9} \le 18 n^{4/9}$ with high probability. 

We now turn to bounding $|Z'|$. This is done in a very similar manner, by partitioning the set $Z'$ into three subsets, $Z'=\bigcup_{\ell=1}^3 Z'_\ell$, where 
$Z'_1$ consists of nodes that woke up at time $\hht+1$ and succeeded in becoming active candidates on their first attempt, 
$Z'_2$ consists of nodes that woke up at time $\hht-2$ and succeeded in becoming active candidates at time $\hht+1$ on their second attempt, and
$Z'_3$ consists of nodes that woke up at time $\hht-5$ and succeeded on their third attempt. 
A similar analysis reveals that $Z'_1$ is no larger than $6n^{1/3}$ with high probability, and each of the other two sets, $Z'_2$ and $Z'_3$, must be no larger than $6n^{4/9}$ with high probability, since otherwise an active candidate would have emerged w.h.p at a time earlier than $\hht+1$.
It follows that $|Z'| \le 18 n^{4/9}$ with high probability. The same analysis yields that also $|Z''| \le 18 n^{4/9}$ with high probability.

It follows that the total number of election messages due to the active candidates in these sets is at most 
$(|Z| + |Z'| + |Z''|)\cdot O(\sqrt{n}\log n) ~\le~ O(n^{4/9}) \cdot O(\sqrt{n}\log n)
~\le~ O(n)$, for sufficiently large $n$.
\end{proof}

This completes the proof of Theorem \ref{the:synch-alg-adv-wakeup}.
\end{proof}

\section{Conclusion}
\label{sec:conc}
In the asynchronous setting, no singularly optimal {\em deterministic} leader election algorithm is known to exist. In contrast, we have shown that using randomization, a singularly optimal algorithm (whose message complexity is asymptotically optimal and whose time complexity is optimal up to a polylogarithmic factor)  can be obtained.

One open question is whether we can improve the time complexity of our message-optimal ($O(n)$) randomized asynchronous algorithm from the current $O(\log^2 n)$ time, to say, $O(\log n)$ time. Also, can we obtain such a (almost) singularly optimal {\em deterministic} algorithm?

Another important question is to find out whether (and when) one can construct time and message-efficient asynchronous algorithms also for {\em general} graphs, in particular algorithms that are (essentially) singularly optimal. In general graphs, that would mean algorithms with
${O}(m)$ (or, at least, $\tilde{O}(m)$) messages and $\tilde{O}(D)$ (or, at least, $\tilde{O}(D)$) 
time. This was shown for the case of synchronous networks in \cite{KPPRT15jacm}.

For synchronous complete networks, an essentially (up to a logarithmic factor) singularly optimal deterministic algorithm was known. We obtained a tightly singularly optimal randomized one (that is, with no extra logarithmic factors).

\bibliographystyle{plainurl}
\bibliography{references}

\end{document}